\documentclass[11pt]{amsart}
\usepackage{graphicx}
\usepackage{stmaryrd}
\usepackage{amssymb}
\usepackage{amsthm}
\usepackage{dsfont}
\usepackage{hyperref}
\usepackage[lite,initials,msc-links]{amsrefs}

\newtheorem{theorem}{Theorem}[section]
\newtheorem{lemma}[theorem]{Lemma}

\newtheorem{corollary}[theorem]{Corollary}
\numberwithin{equation}{section}

\theoremstyle{definition}
\newtheorem*{definition}{Definition}
\theoremstyle{remark}
\newtheorem{remark}{Remark}

\newcommand{\ue}{e}

\newcommand{\ve}{\vec{e}}
\newcommand{\vg}{\vec{g}}

\newcommand{\dG}{\vec{E}(\mathcal{G})}
\newcommand{\Out}{\ensuremath{\text{Out}}}
\newcommand{\KW}{T}
\newcommand{\Id}{\mathrm{Id}}	
\newcommand{\C}{\mathds{C}}
\newcommand{\XX}{\mathcal{X}}
\newcommand{\ZZ}{\mathcal{Z}}
\newcommand{\TrM}{\Lambda}

\newcommand{\RP}{\mathrm{Re}}
\newcommand{\IP}{\mathrm{Im}}
\newcommand{\TrMa}{A}
\newcommand{\TrMb}{B}
\newcommand{\free}{\text{free}}
\newcommand{\dE}{\vec{E}}
\newcommand{\FE}{\Upsilon}
\newcommand{\Blow}{\mathcal{T}_{\text{low}}}
\newcommand{\Bhigh}{\mathcal{T}_{\text{high}}}
\newcommand{\dx}{\vec{x}}

\newcommand{\trace}{\text{tr}}
\newcommand{\FinG}{\mathcal{G}}
\newcommand{\InfG}{\Gamma}

\begin{document}
\date{\today}
\title[Phase transition free regions in the Ising model]{Phase transition free regions in the Ising model via the Kac--Ward operator}

\author{Marcin Lis}
\address{VU University\\
Department of Mathematics\\
De Boelelaan 1081a\\
1081\,HV Amsterdam\\
The Netherlands}
\email{m.lis\,@\,vu.nl}\


\begin{abstract} 
We provide an upper bound on the spectral radius of the Kac--Ward transition matrix for a general planar graph. 
Combined with the Kac--Ward formula for the partition function of the planar Ising model, this allows us to identify regions in the complex plane 
where the free energy density limits are analytic functions of the inverse temperature. The bound turns out to be optimal in the case of
isoradial graphs, i.e.\ it yields criticality of the self-dual Z-invariant coupling constants.
\end{abstract}

\maketitle
\section*{Introduction}
The Ising model, proposed by Lenz~\cite{Lenz}, and solved in one dimension by his student Ising~\cite{Ising}, is one of the most studied models of statistical mechanics.
It was introduced as a model for ferromagnetism with the intention to explain spontaneous magnetization. Ising proved that the one dimensional 
case does not account for the existence of this phenomenon and concluded that the same should hold in higher dimensions. This was 
later disproved by Peierls~\cite{Peierls}, whose, now classical, argument established that in dimensions higher than one the model does exhibit a phase transition in the magnetic behavior. 
The critical point, i.e.,\ the value of the temperature parameter where the phase transition occurs, for the model defined on the two-dimensional square lattice 
was first identified by Kramers and Wannier~\cite{KraWan1} as the fixed point of a certain duality transformation. The first rigorous proof of 
criticality of the self-dual point came together with the exact solution of the two-dimensional model done by Onsager~\cite{Onsager},
who explicitly computed the free energy density and showed that it is not analytic only at this particular value of the temperature. 

Since then, several different methods have been developed to study the two-dimensional Ising model. One of them is the approach of Kac and Ward~\cite{KacWard},
who expressed the partition function of the model in terms of the determinant of what is now called the Kac--Ward operator. This combinatorial in nature idea 
has been so far a source of numerous results about the planar Ising model. The most classical are the (alternative to the solution of 
Onsager and Yang~\cite{Yang}) analytic derivations of the free energy density and magnetization performed by Vdovichenko~\cites{Vdovichenko1,Vdovichenko2}, 
who built on earlier works of Sherman~\cite{Sherman} and Burgoyne~\cite{Burgoyne}.
However, most of the articles concerning the Kac--Ward formula left many details of the 
method unexplained and even contained errors. The first completely rigorous account of this approach seems to be given much later by Dolbilin et al.\
\cites{DolEtAl}. A more recent treatment, presented by Kager, Meester and the author~\cite{KLM}, concentrates on loop expansions
of the Kac--Ward determinants. As a result, the authors not only obtain rigorous proofs of the combinatorial foundations of the approach, 
but also rederive the critical temperature of the Ising model on the square lattice. 
The Kac--Ward determinants also turned out to be the right tool for the computation of the critical point of Ising models defined on planar doubly periodic 
graphs (Cimasoni and Duminil-Copin~\cite{CimDum}). Moreover, Cimasoni~\cite{Cimasoni} showed that the Kac--Ward formula can be
generalized to Ising models defined on surfaces of higher genus. Finally, as pointed out by the author in~\cite{Lis}, 
the Kac--Ward method is intrinsically connected with the discrete holomorphic approach to the Ising model introduced by Smirnov~\cite{Smirnov}.

In this paper, we continue in the spirit of~\cite{KLM}, where the spectral radius and operator norm of the Kac--Ward transition matrices were first considered.
We explicitly compute the operator norm of what we call the conjugated transition matrix defined for a general graph in the plane, 
and hence we provide an upper bound on the spectral radius of the standard Kac--Ward transition matrix. 
Combining this result with the Kac--Ward formula for the high and
low-temperature expansion of the partition function yields domains of parameters of the model where there is no phase transition. 
We will focus only on the analytic properties of the free energy, but our bounds, together with the methods from~\cite{KLM}, 
also allow to identify regions where there is
spontaneous magnetization or exponential decay of the two-point functions. 
The advantage of our approach is that it does not require any form of periodicity of the underlying graph.

Moreover, our results are optimal for the Ising model defined on isoradial graphs with uniformly bounded rhombus angles 
(see condition~\eqref{eq:isoradialcond}), i.e.\ we can conclude that the self-dual 
Z-invariant coupling constants, first considered by Baxter~\cite{Baxter1}, are indeed critical in the classical sense. 
To be more precise, after introducing the inverse temperature parameter~$\beta$ to the corresponding Ising model, 
we show that the thermodynamic limits of the free energy density can have singularities only at $\beta=1$.
The isoradial graphs, or equivalently rhombic lattices, were introduced by Duffin~\cite{Duffin} as potentially the largest family of graphs
where one can do discrete complex analysis. As mentioned in~\cite{ChelkSmir},
this class of graphs seems to be the most general family of graphs where the critical Ising model can be nicely defined, and
it also seems to be the one, where our bounds for the spectral radius and operator norm of the Kac--Ward transition matrix yield the critical point of the Ising model.

The self-dual Z-invariant Ising model has been extensively studied in the mathematics literature.
Chelkak and Smirnov~\cite{ChelkSmir} proved that the associated discrete holomorphic fermion has a universal, conformally invariant 
scaling limit. Boutillier and de Tili{\`e}re~\cites{BoutTil1, BoutTil2} gave a complete description of the corresponding
dimer model, yielding also an alternative proof of Baxter's formula for the critical free energy density. Mercat~\cite{Mercat} defined a 
notion of criticality for discrete Riemann surfaces and investigated its connection with criticality in the Ising model. 
The self-dual Z-invariant Ising model is commonly referred to as critical.
However, criticality in the statistical mechanics sense has been established only in the case of doubly periodic isoradial graphs
(see Example 1.6 of~\cite{CimDum} and the references therein). 
As already mentioned, we extend this result to a wide class of aperiodic isoradial graphs.

This paper is organized as follows: in Section~\ref{sec:Isingresults}, we introduce the Ising model and the notion of phase transition, and we state our main theorem. Section~\ref{sec:KWresults}
defines the Kac--Ward operator and presents its connection to the Ising model. It also contains our results for the Kac--Ward transition matrix. 
The proof of the main theorem is postponed until Section~\ref{sec:freeenergy}.

\section{Results for the Ising model} \label{sec:Isingresults}

\subsection{The Ising model}

Let $\InfG$ be an infinite, planar, simple graph embedded in the complex plane and let $\InfG^*$ be its planar dual.
We assume that both $\InfG$ and $\InfG^*$ have uniformly bounded vertex degrees. One should think of~$\InfG$ as any kind of tiling or discretization of the plane. 
In particular, $\InfG$ can be a regular lattice, or an instance of an isoradial graph (see Section~\ref{sec:isoradialIsing}).
We call a subgraph~$\FinG$ of $\InfG$ a \emph{subtiling} if
there is a collection of faces of~$\InfG$, such that $\FinG$ is the subgraph induced by all edges forming boundaries of these faces.
We define the \emph{boundary} $\partial \FinG $ of $\FinG$ to be the set of vertices of~$\FinG$ 
which lie on the boundary of at least one face which is not in the defining collection of faces. 
For a simple graph~$\FinG$ embedded in the complex plane, we will write $V(\FinG)$ for the set 
of vertices of $\FinG$, which we identify with the corresponding complex numbers. By $E(\FinG)$ we will denote the set of edges which are represented by unordered pairs of vertices.

Let  $J=(J_{\ue})_{\ue \in E(\InfG)}$ be a system of \emph{ferromagnetic}, i.e.\ positive, \emph{coupling constants} on the edges of $\InfG$.
For each finite subtiling $\FinG$, we will consider an \emph{Ising model} on $\FinG$ defined by $J$ and the \emph{inverse temperature} parameter $\beta$. Borrowing the notation from \cite{KLM}, let 
\[ 
\Omega_{\FinG}^{\free}=\{-1,+1\}^{V(\FinG)} \quad \text{and} \quad \Omega_{\FinG}^+=\{ \sigma \in \Omega_{\FinG}^{\free} : \sigma_z =+1 \text{ if } z \in\partial \FinG   \}
\]
 be the spaces of \emph{spin configurations} with \emph{free} and \emph{positive boundary conditions}. The Ising model with $\Box$ boundary conditions ($\Box \in \{\free, +\}$) 
is defined by a probability measure on $\Omega^{\Box}_{\FinG}$ given by 
\begin{align*}
\mathbf{P}^{\Box}_{\FinG,\beta}(\sigma) = \frac{1}{\ZZ^{\Box}_{\FinG}(\beta) }\prod_{ \{z,w\} \in E(\FinG)} \exp \big(\beta J_{\{z,w\}} \sigma_z \sigma_w \big),\qquad \sigma \in \Omega^{\Box}, 
\end{align*} 
where the normalizing factor
\begin{align*}
\ZZ^{\Box}_{\FinG}(\beta) = \sum_{\sigma \in \Omega^{\Box}} \prod_{ \{z,w\} \in E(\FinG)} \exp \big(\beta J_{\{z,w\}} \sigma_z \sigma_w \big)
\end{align*} 
is called the \emph{partition function}.

Throughout the paper, we will make a natural assumption on the coupling constants, namely we will require
that there exist numbers $m$ and $M$, such that for all $\ue \in E(\InfG)$,
\begin{align} \label{eq:regularcoupling} 
0 < m \leq J_{\ue} \leq M < \infty.
\end{align} 

\subsection{Phase transition}
An object of interest in statistical physics is the \emph{free energy density} (or \emph{free energy per site}) defined by
\begin{align*}
f^{\Box}_{\FinG}(\beta) = - \frac{\ln \ZZ^{\Box}_{\FinG}(\beta)}{\beta|V(\FinG)|} . 
\end{align*}
It is clear that the free energy density is an analytic function of the inverse temperature $\beta \in (0,\infty)$ for every finite subtiling $\FinG$.
However, when $\FinG$ approaches 
$\InfG$, or more generally, some infinite subgraph of $\InfG$ (this is called taking a \emph{thermodynamic limit}), the
limiting function can have a \emph{critical point}, i.e.\ a particular value of $\beta$ where it is not analytic. The existence of such a point indicates
that the system undergoes a phase transition when one varies $\beta$ through the critical value. This is a universal way of looking at the phenomenon of phase transition since 
it can be applied to any model of statistical mechanics.

Another approach, and perhaps a more natural one in the setting of the Ising model, is to investigate the magnetic behavior of the system. 
To this end, one defines the \emph{spin correlation functions}, i.e.\ the expectations of products of the \emph{spin variables} taken with 
respect to the Ising probability measure. The simplest cases are the \emph{one} and \emph{two-point functions}
\begin{align*}
\langle \sigma_{z} \rangle^{\Box}_{\FinG,\beta} = \sum_{\sigma \in \Omega^{\Box}_{\FinG}} \sigma_{z} \mathbf{P}^{\Box}_{\FinG,\beta}(\sigma), 
\quad \langle \sigma_{z}\sigma_{w} \rangle^{\Box}_{\FinG,\beta} = \sum_{\sigma \in \Omega^{\Box}_{\FinG}} \sigma_{z}\sigma_{w} \mathbf{P}^{\Box}_{\FinG,\beta}(\sigma), \quad z,w \in V(\FinG).
\end{align*}
Since the model is ferromagnetic, and due to the effect of positive boundary conditions, 
the corresponding one-point function $\langle \sigma_{z} \rangle^{+}_{\FinG,\beta}$ is strictly positive for all finite subtilings $\FinG$ and for all $\beta$. 
In other words, in finite volume, the spins prefer the~$+1$ state at all temperatures. 
However, when $\FinG$ approaches~$\InfG$, the boundary moves further and further away
and, at temperatures high enough, its influence on a particular spin vanishes. As a result, the limiting one-point function equals zero and the spin equally likely occupies the~$+1$ and~$-1$ state.
On the other hand, this does not happen at low temperatures, i.e.\ if $\beta$ is sufficiently large, then $\langle \sigma_{z} \rangle^{+}_{\FinG,\beta}$ 
stays bounded away from zero uniformly in~$\FinG$.
This means that the effect of positive boundary conditions is carried through all length scales and 
there is \emph{spontaneous magnetization}. In this approach, the critical point is the value of~$\beta$, which separates the regions with and without spontaneous magnetization.
In some cases, it is more convenient to investigate the behavior of the two-point functions.
Here, one also discerns two different non-critical cases: either the system is \emph{disordered}, i.e.\ 
the thermodynamic limits of the two point functions decay exponentially fast to zero with the graph distance between 
$z$ and $w$ going to infinity, or the system is \emph{ordered}, which means that the limiting two-point functions stay bounded away from zero uniformly in $z$ and $w$. 
For periodic Ising models, the critical point defined as the value of $\beta$ which separates these two regimes is the same as the critical point defined via spontaneous magnetization
(see Theorem 1 in \cite{ABF} and the references therein). In particular, the system exhibits long-range ferromagnetic order if and only if there is spontaneous magnetization.

Property~\eqref{eq:regularcoupling}, 
together with the conditions we imposed on~$\InfG$ and~$\InfG^*$, is enough for the existence of a phase transition in terms of spontaneous magnetization and the behavior of the two-point functions. 
This is a consequence of the classical arguments of Peierls~\cite{Peierls} and Fisher~\cite{Fisher}. 
In this paper, we will only focus on the phase transition in the analytic behavior of the free energy density limits, but our results for the Kac--Ward operator
can be also used in the setting of the magnetic phase transition (see Section~\ref{sec:magnetic}). 

\subsection{The main result}
Let $\dE(\FinG)$ be the set of directed edges of~$\FinG$ which are the ordered pairs of vertices. For a directed edge $\ve=(z,w)$, we define its \emph{reversion} by $-\ve=(w,z)$ and we 
obtain the undirected version by dropping the arrow from the notation, i.e.\ $\ue=\{z,w\}$. If $z$ is a vertex, then we write $\Out_{\FinG}(z)=\{ (z',w')\in \dE(\FinG): z'=z\}$ 
for the set of edges emanating from~$z$.

Let $\dx$ and $x$ be systems of nonzero complex weights on the directed and undirected edges of $\FinG$ respectively. 
We call $\dx$ \emph{(Kac--Ward) contractive} if 
\begin{align} \label{eq:contractive}
\sum_{\ve \in \Out_{\FinG}(z)} \arctan |\dx_{\ve}|^2 \leq \frac{\pi}{2} \qquad \text{for all} \ z \in V(\FinG),
\end{align}
and we say that $x$ \emph{factorizes} to~$\dx$ if
\begin{align} \label{eq:factorizes}
x_{\ue} = \dx_{\ve} \dx_{-\ve} \qquad \text{for all} \ \ve \in \dE(\FinG).
\end{align}
For the origin of condition \eqref{eq:contractive}, see Corollary \ref{cor:contraction}.

In the context of the Ising model, two particular systems of edge weights will be important, namely the so called \emph{high} and \emph{low-temperature} weights given by
\begin{align*}
\tanh \beta J = \big(\tanh \beta J_{\ue}\big)_{\ue \in E(\InfG)} \quad \text{and} \quad  \exp(-2 \beta J) = \big(\exp(-2 \beta J_{\ue})\big)_{\ue^* \in E(\InfG^*)}.
\end{align*}

\begin{definition}
We say that the coupling constants satisfy the \emph{high-temperature} condition if $\tanh J$ factorizes to a contractive system of weights on the directed edges of $\InfG$, and we say that they satisfy the 
\emph{low-temperature} condition if $\exp(-2 J)$ factorizes to a contractive system of weights on the directed edges of $\InfG^*$.
\end{definition}
Let
\begin{align*}
\FE_{\Box} = \big\{ f^{\Box}_{\FinG}: \text{$\FinG$ is a finite subtiling of $\InfG$} \big \}
\end{align*} 
be the family of all free energy densities with $\Box$ boundary conditions, and let~$\overline{\FE}_{\Box}$ be its closure in the topology of pointwise convergence on $(0,\infty)$. 
Note that~$\overline{\FE}_{\Box}$ contains all thermodynamic limits and can also contain other types of accumulation points of $\FE_{\Box}$.
Using the definition of $\ZZ^{\Box}_{\FinG}$, it is not difficult to prove that, under condition \eqref{eq:regularcoupling}, $\FE_{\Box}$ is uniformly bounded and 
equicontinuous on compact subsets of $(0,\infty)$. In particular, all sequences 
in $\FE_{\Box}$ which converge pointwise, converge uniformly on compact sets, 
and therefore all functions in~$\overline{\FE}_{\Box}$ are continuous on $(0,\infty)$. However, this is not enough to conclude analyticity of the limiting functions, and indeed, critical points do arise. 

In this paper, we show that, if the coupling constants satisfy the high-temperature condition, then all functions in $\FE_{\free}$ 
can be extended analytically to a complex domain

\begin{align*}
\Bhigh & =\Big\{\beta :\ 0< \RP \beta<1,\  2M|\IP \beta| < \frac{\pi}{2},\   \frac{\cosh(2 m\RP \beta )}{\cosh(2m) \cos (2M \IP \beta) } < 1 \Big\}
\end{align*}
which we call the \emph{high-temperature regime}. Note that $(0,1) \subset \Bhigh$.
Similarly we prove that, if the coupling constants satisfy the low-temperature condition, then all functions in $\FE_{+}$ can be extended to analytic functions on 
\begin{align*}
\Blow=\{\beta:  1 < \RP \beta \}
\end{align*}
which we call the \emph{low-temperature regime}. Moreover, we show that $\FE_{\Box}$ is uniformly bounded on compact subsets of the
corresponding regimes. 

\begin{figure} 
		\begin{center}
			\includegraphics{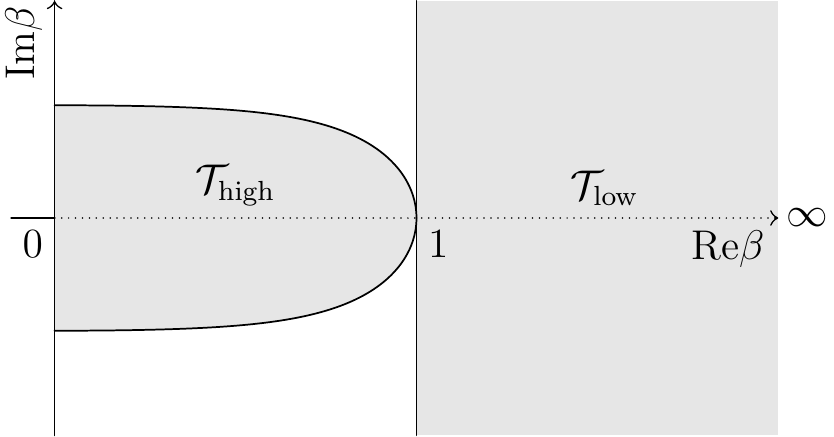}
		\end{center}
		\caption{The high and low-temperature regimes.}
		\label{fig:regimes}
\end{figure}

For complex analytic functions, this is enough to conclude that all pointwise limits are also complex analytic.
More precisely, let $D$ be a complex domain and let $E\subset D$ have an accumulation point in $D$.
The Vitali-Porter theorem (see~\cite{Shiff}*{\S 2.4}) states that if a sequence of holomorphic functions 
defined on $D$ converges pointwise on $E$, and is uniformly bounded
on compact subsets of $D$, then it converges uniformly on compact subsets of $D$ and the limiting function is holomorphic.
In our context, the role of the domain $D$ is played by the high and low-temperature regimes, and 
$E$ is the intersection of the given regime with the positive real numbers.
 
In other words, under the high and low-temperature conditions on the coupling constants, the high and low-temperature regimes
are free of phase transition in terms of analyticity of the thermodynamic limits of the free energy density.
This is summarized in the following theorem:
\begin{theorem} \label{thm:FreeEnergy}
If the coupling constants satisfy
\begin{itemize} 
\item[(i)] the high-temperature condition, then all functions in $\FE_{\free}$ extend analytically to $\Bhigh$, and $\FE_{\free}$ is 
uniformly bounded on compact subsets of~$\Bhigh$. As a consequence, all functions in $\overline{\FE}_{\free}$ are analytic on $\Bhigh$, and in particular on $(0,1)$.
\item[(ii)] the low-temperature condition, then all functions in $\FE_+$ extend analytically to $\Blow$, and $\FE_+$ is uniformly bounded on
 compact subsets of~$\Blow$. As a consequence, all functions in $\overline{\FE}_{+}$ are analytic on $\Blow$, and in particular on $(1,\infty)$.
\end{itemize}
\end{theorem}
The proof of this theorem is provided in Section \ref{sec:freeenergy}.
Its main ingredients are the Kac--Ward formula for the partition function of the Ising model (see Theorem~\ref{thm:expansions})
and the bound on the spectral radius of the the Kac--Ward transition matrix given in Theorem~\ref{thm:boundedradius}.

In most of the applications, the role of boundary conditions is immaterial for the thermodynamic limit of the free energy density. Indeed, 
it is not hard to prove that whenever $|\partial \FinG|/|V(\FinG)|$ is small, 
then for $\beta\in(0,\infty)$, $f^{\free}_{\FinG}(\beta)$ and $f^+_{\FinG}(\beta)$ are close to each other 
(and also to any other free energy density function defined for other types of boundary conditions on $\FinG$). Hence, limits of the free energy density taken along sequences, where the above ratio approaches zero, are the same for all boundary conditions.
In this paper, we consider the free and positive boundary conditions since in these cases, the partition function of the model is given in terms of the determinant of the Kac--Ward operator. 
Thus, one can use properties of the operator itself to derive results for the free energy density.

\subsection{The isoradial case} \label{sec:isoradialIsing}

\begin{figure} 
		\begin{center}
			\includegraphics{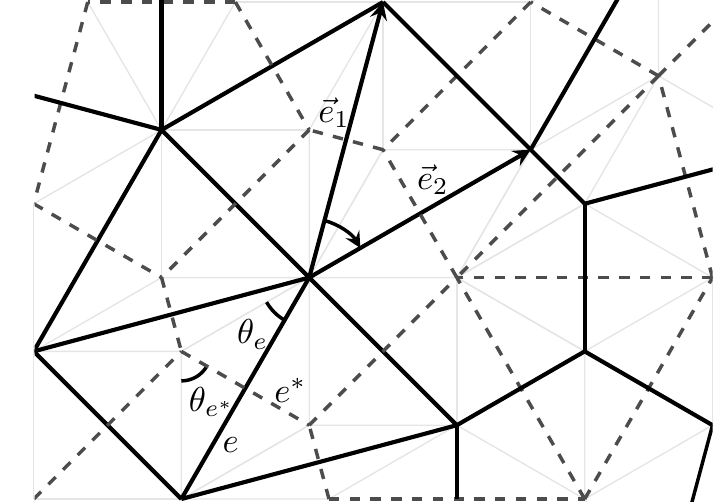}
		\end{center}
		\caption{Local geometry of an isoradial graph and its dual. The underlying rhombic lattice is drawn in pale lines. The directed arc marks the turning angle $\angle(\ve_1,\ve_2)$.}
		\label{fig:isoradial}
\end{figure}

Assume that $\InfG$ is an isoradial graph, i.e.\ all its faces can be inscribed in circles with a common radius,
and all the circumcenters lie within the corresponding faces. An equivalent characterization says
that $\InfG$ and $\InfG^*$ can be simultaneously embedded in the plane in such a way, that each pair of mutually dual edges
forms diagonals of a rhombus. The roles of $\InfG$ and $\InfG^*$ are therefore symmetric and the dual graph is also isoradial. 
The simplest cases of isoradial graphs are the regular lattices: the square, triangular and hexagonal lattice.

One assigns to each edge $\ue$ the interior angle $\theta_{\ue}$ that $\ue$ creates with any side of the associated rhombus (see Figure~\ref{fig:isoradial}).
Note that $\theta_{\ue}+\theta_{\ue^*}= \pi/2$.
There is a particular geometric choice of the coupling constants given by
\begin{align} \label{eq:zinvariant}
\tanh J_{\ue} = \tan (\theta_{\ue}/2), \quad \text{or equivalently,} \quad \exp(-2 J_{\ue})= \tan (\theta_{\ue^*}/2).
\end{align} 
These coupling constants were first considered by Baxter~\cites{Baxter1}.
We will refer to them as the \emph{self-dual Z-invariant} coupling constants since these are the only coupling constants that 
make the Ising model invariant under the star-triangle transformation, and also satisfy the above generalized Kramers-Wannier self-duality \eqref{eq:zinvariant}. 
For more details on their origin, see \cites{BoutTil1,Baxter2}.

Observe that in this setting, condition \eqref{eq:regularcoupling} is equivalent to the existence of constants 
$k$ and $K$, such that for all $\ue \in E(\InfG)$,
\begin{align} \label{eq:isoradialcond}
0 < k \leq \theta_{\ue} \leq K < \pi.
\end{align}
This means that the associated rhombi have a positive minimal area, and also gives a uniform
bound on the maximal degree of~$\InfG$ and~$\InfG^*$.

The next corollary states that, for the Ising model defined by the above coupling constants, the only possible point of phase transition in the analytic behavior of the free energy density 
is~$\beta =1$.
\begin{corollary} \label{cor:isoradial}
Let $\InfG$ be an isoradial graph satisfying condition~\eqref{eq:isoradialcond}.
Consider Ising models defined by the self-dual Z-invariant coupling constants on finite subtilings of~$\InfG$.
Then, all functions in $\overline{\FE}_{\free}$ are analytic on~$(0,1)$, and all functions in $\overline{\FE}_{+}$ are analytic on $(1,\infty)$.
\end{corollary}
\begin{proof}
By \eqref{eq:zinvariant} and the fact that the angles $\theta$ sum up to $\pi$ around each vertex of $\InfG$ and $\InfG^*$,
 the self-dual Z-invariant coupling constants simultaneously satisfy the 
high and low-temperature condition. Indeed, the contractive weight systems on the directed edges are given by 
$\dx_{\ve} = \sqrt{\tan (\theta_{\ue}/2)}$.
The claim follows therefore
from Theorem \ref{thm:FreeEnergy}.
\end{proof}
Note, that in this case, the inequalities in \eqref{eq:contractive} become equalities.

\subsection{Implications for the magnetic phase transition} \label{sec:magnetic}

Recently~\cite{KLM}, the Kac--Ward operator and the signed weights it induces on the closed non-backtracking walks in a graph
were used to rederive the critical temperature of the homogeneous Ising model on the square lattice. 
It was done both in terms of analyticity of the free energy density limit and the change in behavior of the
one and two-point functions.
The methods used there to analyse the correlation functions work also for general planar graphs under
some slight regularity constraints. To be more precise, the proof of Theorem~1.4 in~\cite{KLM} which gives the existence of spontaneous 
magnetization, uses the fact that 
appropriate Kac--Ward transition matrices have spectral radius smaller than one and that the dual graph (which is $\InfG^*$ in our setup)
has subexponential growth of volume, i.e.\ the volume of balls in graph distance grows subexponentially with the radius.
This condition is, for instance, satisfied by all isoradial graphs where~\eqref{eq:isoradialcond} holds true.
On the other hand, Theorem~1.6 and Corollary~1.7 from \cite{KLM}, which yield exponential decay of the two-point functions,
use the fact that the operator norm of appropriate Kac--Ward matrices is smaller than one.

The bounds that are stated in Section~\ref{sec:KWresults} allow
to generalize the above results to arbitrary planar graphs, i.e.\ together with the methods from~\cite{KLM} they
 provide regions of parameters~$J$ and~$\beta$
where there is spontaneous magnetization or exponential decay of the two-point functions. These regions coincide with those in 
Theorem~\ref{thm:FreeEnergy} (one can analytically extend the correlation functions to the high and low-temperature regime),
that is, if the coupling constants satisfy the low-temperature condition, then there is spontaneous magnetization on~$\Blow$,
and if they satisfy the high-temperature condition, then there is exponential decay of the two-point functions on~$\Bhigh$.
In particular, our bounds together with the methods developed in~\cite{KLM} prove that the self-dual Z-invariant weights are critical in the sense of magnetic phase transition.

We would also like to point out that the arguments, which are used in~\cite{KLM} to conclude analyticity of the free energy density limit,
do not work for general graphs since they rely on periodicity of the square lattice. 
This is why, in this paper, we go into details of this aspect of phase transition and we do not focus on the magnetic behavior of the model.

\section{Results for the Kac--Ward operator} \label{sec:KWresults}

\subsection{The Kac--Ward operator and the Ising model}

Let $\FinG$ be a finite simple graph embedded in the plane.
For a directed edge $\ve=(z,w)$, we define its \emph{tail} $t(\ve) =z$ and \emph{head} $h(\ve)=w$.
For $\ve,\vg \in \dE(\FinG)$, let
\begin{align} \label{eq:defangle}
\angle(\ve,\vg)= \text{Arg}\Big(\frac{h(\vg)-t(\vg)}{h(\ve)-t(\ve)}\Big) \in (-\pi,\pi]
\end{align}
be the \emph{turning angle} from  $\ve$ to $\vg$ (see Figure \ref{fig:isoradial}). 
The \emph{transition matrix} for~$\FinG$ and the weight system $x$ is given by
\begin{align} \label{eq:transitionmatrix1}
\TrM_{\ve,\vg}(x) = \begin{cases}
		x_{\ue} e^{\frac{i}{2}\angle(\ve,\vg)}
		& \text{if $h(\ve) = t(\vg)$ and $\vg \neq -\ve$}; \\
		0 & \text{otherwise},
	\end{cases}
\end{align}
where $\ve,\vg \in \dE(\FinG)$.
To each $\ve \in \dE(\FinG)$ we attach a copy of the complex numbers denoted by $\C_{\ve}$ and
we define a complex vector space 
\begin{align*}
\XX = \prod_{\ve \in \dE(\FinG)} \C_{\ve}.
\end{align*}
We identify $\TrM(x)$ with the automorphism of $\XX$ it defines via matrix multiplication. 
The \emph{Kac--Ward operator} for $\FinG$ and the weight system $x$ is the automorphism of $\XX$ given by
\begin{align*}
 \KW(x) = \Id -\TrM(x),
\end{align*}
where $\Id$ is the identity on $\XX$. When necessary, we will use subscripts to express the fact that the above operators depend on the underlying graph~$\FinG$.

If $\FinG$ is a finite subtiling of $\InfG$, then we will denote by $\FinG^*$ the subgraph of~$\InfG^*$ whose edge set
consists of all dual edges $e^*$, such that at least one of the endpoints of $\ue$ belongs to~$V(\FinG) \setminus \partial \FinG$. One can see that~$\FinG^*$ 
is a subtiling of~$\InfG^*$ whose defining set of dual faces is given by the vertices from $V(\FinG) \setminus \partial \FinG$.
We will call it the \emph{dual subtiling} of $\FinG$.

We say that a graph is \emph{even} if all its vertices have even degree.
There are two classical methods of representing the partition function of the Ising model on $\FinG$
as a weighted sum over all even subgraphs of $\FinG$ or $\FinG^{*}$. 
The first one, called the \emph{low-temperature expansion}, involves a bijective mapping between the spin configurations
with positive boundary conditions and the collection of even subgraphs of $\FinG^{*}$. The graph associated with a spin configuration is composed 
of these dual edges, whose corresponding primal edge has two opposite values of spins assigned to its endpoints. Hence,
the resulting even subgraph forms an interface between the clusters of positive and negative spins in the configuration.
In this expansion, each even graph is given a weight which is proportional to the product of the low-temperature
edge weights $\exp(-2 \beta J)$ taken over all edges in the graph. 
The second method is called the \emph{high-temperature expansion} and it is a way of expressing 
the partition function of the Ising model with free boundary conditions as a sum over all even subgraphs of $\FinG$.
Similarly, it assigns to each even subgraph a product weight composed of factors given by 
the high-temperature weight system $\tanh \beta J$. However, unlike in the low-temperature case, the even subgraphs
do not have a geometrical interpretation in terms of the spin variables. 
The weighted sums arising in both of these expansions are called the \emph{even subgraph generating functions}.

The Kac--Ward formula expresses the square of an even subgraph generating function
as the determinant of a Kac--Ward matrix with an appropriate edge weight system. 
The combined result of the high and low-temperature expansion together with the Kac--Ward formula is stated
in the next theorem. 
Here, we assume that the edges of $\FinG$ (and also $\FinG^{*}$) are embedded as straight line segments 
which do not intersect.
For the origin of this condition, a detailed account of the high and low-temperature expansion, and the proof of the following theorem, see~\cite{KLM}.

\begin{theorem} \label{thm:expansions}
For all choices of the coupling constants $J$ and all $\beta$ with $\RP \beta > 0$,
\begin{align*}
\quad \text{\textnormal{(i)}}& \quad \big(\ZZ^{\free}_{\FinG}(\beta) \big)^2 = 2^{2|V(\FinG)|} 
\Big( \prod_{\ue \in E(\FinG)} \cosh^2(\beta J_{\ue})\Big)\det\big[\KW_{\FinG}(\tanh \beta J)\big],\\
\quad \text{\textnormal{(ii)}}& \quad \big(\ZZ^{+}_{\FinG}(\beta) \big)^2 = \exp \Big(2\beta \sum_{\ue \in E(\FinG)} 
J_{\ue}\Big) \det\big[\KW_{\FinG^{*}}(\exp(-2 \beta J))\big].
\end{align*}
\end{theorem}
Note that the condition $\RP \beta >0$ is needed only for the weight system $\tanh \beta J$ to be well defined.

The determinant of the Kac--Ward matrix is the characteristic polynomial of the transition matrix evaluated at one:
\begin{align*}
\det \KW = \det(\Id -\TrM) = \prod_{k=1}^{2n}(1- \lambda_k),
\end{align*}
where $n$ is the number of edges of $\FinG$, and $\lambda_k$, $k\in \{1,2,\ldots, 2n \}$, are the eigenvalues of $\TrM$. 
Recall that we want to extend the free energy density functions to domains in the complex plane.  
The free energy density is given by the logarithm of the partition function, and the square of the partition function is proportional to the above product
involving eigenvalues of the transition matrix.
In this situation, it is natural to use the power series expansion of the logarithm around one:
\begin{align*}
\ln (1- \lambda) = -\sum_{r=1}^{\infty} \lambda^r/r, \qquad |\lambda| <1.
\end{align*}
This series is convergent whenever $\lambda$ stays within the unit disc, and hence we should require that the spectral radius of the
transition matrix is bounded from above by one. The next section is devoted to providing the necessary estimates.

\subsection{Bounds on the spectral radius and operator norm}
In this paper we will make use of transition matrices conjugated by diagonal matrices of a certain type:
if $x$ factorizes to $\dx$ (see~\eqref{eq:factorizes}), then we define the \emph{conjugated transition matrix} by
\begin{align*}
 \TrM(\dx)= D^{-1}(\dx) \TrM(x) D(\dx),
\end{align*}
where $D(\dx)$ is the diagonal matrix satisfying $D_{\ve,\ve}(\dx)=\dx_{\ve}$ for all $\ve \in \dG$. 
The resulting transition matrix takes the following form:
\begin{align} \label{eq:transitionmatrix2}
\TrM_{\ve,\vg}(\dx)= \begin{cases}
		\dx_{-\ve}\dx_{\vg} e^{\frac{i}{2}\angle(\ve,\vg)}
		& \text{if $h(\ve) = t(\vg)$ and $\vg \neq -\ve$}; \\
		0 & \text{otherwise}.
	\end{cases}
\end{align}
This matrix is similar to the standard transition matrix, and in particular has the same spectrum. 
Moreover, it turns out that one can explicitly compute its operator norm.

To this end, let us make some additional observations. 
For a square matrix~$A$, let~$\| A \|$ be its operator norm induced by the Euclidean norm, 
and let $\rho(A)$ be its spectral radius.
Note that there is a natural involutive automorphism $P$ of $\XX$ induced by the map $\ve \mapsto -\ve$, i.e.\ 
the automorphism which assigns to each complex number in $\C_{\ve}$ the same complex number in~$\C_{-\ve}$. 
Fix $\dx$ and let $\TrMa= P \TrM(\dx)$. Observe that $\| \TrMa \| =\|  \TrM(\dx) \|$ since $P$ is an isometry.  
Moreover, the operator norm of $\TrMa$ depends only on the absolute values of $\dx$. Indeed, if
\begin{align*}
\TrMb  =  D(\vec{u})  \TrMa D(\vec{u}), \qquad \text{where} \quad \vec{u}_{\ve} = |\dx_{\ve}|/\dx_{\ve},
\end{align*} 
then $\TrMb$ is given by the matrix
\begin{align} \label{eq:transitionmatrixb}
\TrMb_{\ve,\vg} = \begin{cases}
		|\dx_{\ve}\dx_{\vg}| e^{\frac{i}{2}\angle(-\ve,\vg)}
		& \text{if $t(\ve) = t(\vg)$ and $\vg \neq \ve$}; \\
		0 & \text{otherwise},
	\end{cases}
\end{align}
and $\| \TrMb\| = \|\TrMa\|$ since $D(\vec{u})$ is an isometry.

Note that $\XX$ can be decomposed as
\begin{align*}
\XX = \prod_{z \in V(\FinG)} \XX^z, \qquad \text{where} \quad \XX^z = \prod_{\ve \in \Out_{\FinG}(z)} \C_{\ve}.
\end{align*}
One can see from \eqref{eq:transitionmatrixb} that $\TrMb$ gives a nonzero transition weight only between two edges sharing the same 
tail $z$. In other words, $\TrMb$ maps $\XX^z$ to itself and therefore is block-diagonal, that is
\begin{align*}
\TrMb = \prod_{z \in V(\FinG)} \TrMb^z,
\end{align*}
where $\TrMb^z : \XX^z \rightarrow \XX^z$
is the restriction of $\TrMb$ to the space $\XX^z$. Moreover, the angles satisfy
\begin{align} \label{eq:anglereflection}
\angle(-\ve,\vg) = - \angle(-\vg,\ve) \qquad \text{for} \quad \ve \neq \vg,
\end{align}  
and hence $\TrMb$ is Hermitian, i.e.\ $\TrMb_{\ve,\vg} = \overline{\TrMb_{\vg,\ve}}$.
Combining these two properties and the fact that the operator norm of a Hermitian matrix is given by its spectral radius, we arrive at 
the identity:
\begin{align} \label{eq:blockhermitian}
\| \TrMb \| = \rho(\TrMb)= \max_{z \in V(\FinG)} \rho(\TrMb^z).
\end{align} 

It turns out that the characteristic polynomial of $\TrMb^z$ is easily expressible in terms of the weight vector $\dx$:
\begin{lemma} \label{lem:charpol}
For any real $t$ and any vertex $z$,
\begin{align*}
\det(t\Id- \TrMb^z) = \RP \Big(\prod_{\ve \in \text{\textnormal{Out}}_{\FinG}(z)} (t + i|\dx_{\ve}|^2) \Big),
\end{align*}
where $\Id$ is the identity on $\XX^z$.
\end{lemma}
\begin{proof}
The proof is by induction on the degree of $z$.
One can easily check that the statement is true for all vertices of degree one or two.
Now suppose that it is true for all vertices of degree at most $n \geq 2$. Let~$z$ be a vertex of degree $n+1$ and
let $\ve_1,\ve_2, \ldots, \ve_{n+1}$ be a counterclockwise ordering of the edges of $\Out_{\FinG}(z)$. Consider the matrix $S=t\Id - \TrMb^z$ 
with columns and rows ordered accordingly. 
Note that for all $\vg \in \Out_{\FinG}(z)$ different from $\ve_1$ and~$\ve_2$,
\begin{align*} 
\angle(\vg,\ve_1)+  \angle(\ve_1, \ve_2) + \angle(\ve_2,\vg) = 0 \ (\text{mod} \ 2 \pi).
\end{align*}
Also observe that, for geometric reasons, at least two of the above angles are positive.
Combining this together with the fact that $\text{Arg}(w) = \text{Arg}(-w) \pm \pi$ for any complex~$w$,
and that the angles are between~$-\pi$ and~$\pi$, yields
\begin{align} \label{eq:charpol1}
\angle(-\ve_1,\vg) = \angle(-\ve_2,\vg) + \angle(-\ve_1, \ve_2) + \pi.
\end{align}
This identity guarantees that every two consecutive rows and columns of $S$ are ``almost proportional'' to each other.

To be more precise,
we first subtract from the first row of $S$, the second row multiplied by $i e^{\frac{i}{2} \angle(-\ve_1, \ve_2)} |\dx_{\ve_1}| /|\dx_{\ve_2}| $. Then, we subtract from the first 
column the second one multiplied by $-i e^{-\frac{i}{2} \angle(-\ve_1, \ve_2)}  |\dx_{\ve_1}| /|\dx_{\ve_2}|$. The resulting matrix has the same determinant as $S$. By the definition of~$\TrMb^z$, 
\eqref{eq:anglereflection} and~\eqref{eq:charpol1},
\begin{align*} \det S =
\det \begin{pmatrix}
 a & b & 0 & 0 &\cdots \\ 
\overline{b} & t & -\TrMb^z_{\ve_2,\ve_3} & -\TrMb^z_{\ve_2,\ve_4} & \cdots \\
0 & -\overline{\TrMb^z}_{\ve_2,\ve_3} & t &  -\TrMb^z_{\ve_3,\ve_4} & \cdots \\
0 & -\overline{\TrMb^z}_{\ve_2,\ve_4} & -\overline{\TrMb^z}_{\ve_3,\ve_4} & t & \cdots \\
\vdots & \vdots & \vdots & \vdots &\ddots
\end{pmatrix}
\end{align*}
where $a=t\big(1 + |\dx_{\ve_1}|^2 /|\dx_{\ve_2}|^2 \big)$ and $b = -e^{\frac{i}{2} \angle(-\ve_1,\ve_2)} \big(it|\dx_{\ve_1}| /|\dx_{\ve_2}|+|\dx_{\ve_1}\dx_{\ve_2}|\big)$.
Let $S_1$ be the matrix resulting from removing from $S$ the first column and the first row, and let $S_2$ be the matrix,
where the first two rows and the first two columns of $S$ are removed. By the induction hypothesis, $\det S_1 = \RP\big((t+i|\dx_{\ve_2}|^2)\vartheta\big)$
and $S_2= \RP \vartheta$, where $\vartheta = \prod_{\vg \in \Out_{\FinG}(z) \setminus \{\ve_1, \ve_2 \}}(t+i|\dx_{\vg}|^2)$. Expanding the determinant,
we get 
\begin{align*}
\det S &= a \det S_1 - b \overline{b} \det S_2 \\ 
		&= t\big(1 +  |\dx_{\ve_1}|^2 /|\dx_{\ve_2}|^2 \big) \RP\big((t+i|\dx_{\ve_2}|^2)\vartheta\big) \\ 
		&- \big(|\dx_{\ve_1}|^2|\dx_{\ve_2}|^2 + t^2 |\dx_{\ve_1}|^2/|\dx_{\ve_2}|^2\big) \RP\vartheta \\
		&= \RP \big( ( t + i |\dx_{\ve_1}|^2) (t+ i |\dx_{\ve_2}|^2) \vartheta \big).
\end{align*}
The last equality follows since both sides are real linear in $\vartheta$, and one can check that it holds true for $\vartheta=1,i$.
\end{proof}
For $z \in V(\FinG)$, we define $\xi^z(\dx)$ to be the unique solution in~$s$ of the equation 
\begin{align} \label{eq:solutiondef}
\sum_{\ve \in \text{\textnormal{Out}}_{\FinG}(z)}\arctan \big(|\dx_{\ve}|^2/s\big)= \frac{\pi}{2}.
\end{align}
As a corollary we obtain the following result:
\begin{corollary} \label{cor:HTrMrho} 
 $\rho(\TrMb^z) = \xi^z(\dx)$.
\end{corollary}
\begin{proof}
Since $\TrMb^z$ is Hermitian, it has a real spectrum. 
By Lemma~\ref{lem:charpol}, the characteristic polynomial of~$\TrMb^z$ at a nonzero real number~$t$ is given by
\begin{align*}
t^{|\Out_{\FinG}(z)|} \Big(\prod_{\ve \in \Out_{\FinG}(z)} \cos\big( \arctan (|\dx_{\ve}|^2/t) \big)\Big)^{-1}  \cos\Big( \sum_{\ve \in \Out_{\FinG}(z)} \arctan(|\dx_{\ve}|^2/t)\Big) .
\end{align*}
This expression vanishes only when the last cosine term is zero. The largest in modulus values of $t$ for which this happens are equal to $\pm \xi_{\FinG}^z(\dx)$.
\end{proof}

We can now compute the operator norm of the conjugated transition matrix $\TrM(\dx)$.
The following result is the main tool in our considerations:
\begin{lemma} \label{lem:normbound}
\begin{align*}
\| \TrM(\dx)\| = \max_{z \in V(\FinG)} \xi^z(\dx).
\end{align*}
\end{lemma} 
\begin{proof}
It follows from  the fact that $\|\TrM(\dx)\| =\|\TrMb \|$, identity \eqref{eq:blockhermitian}, and Corollary~\ref{cor:HTrMrho}.
\end{proof}

Note that the operator norm depends only on the absolute values of $\dx$.
One can rephrase this result as follows:
\begin{corollary} \label{cor:contraction}
$\| \TrM(\dx) \| \leq s$ if and only if
\begin{align*} 
 \sum_{\ve \in \text{\textnormal{Out}}_{\FinG}(z)}\arctan \big(|\dx_{\ve}|^2/s\big) \leq \frac{\pi}{2} \qquad \text{for all} \ z \in V(\FinG).
\end{align*}
\end{corollary}

We say that an operator is a \emph{contraction} if its operator norm is smaller or equal one, and hence the name of condition \eqref{eq:contractive}.
Since the operator norm bounds the spectral radius from above, we obtain the following corollary:
\begin{corollary} \label{cor:spectralbound}
If $x$ factorizes to $\dx$, then
\begin{align*} 
\rho(\TrM(x) ) \leq  \max_{z \in V(\FinG)} \xi^z(\dx).
\end{align*}
\end{corollary}
This inequality is preserved when one takes the infimum over all factorizations of the weight system $x$. One can check 
that the spectral radius of the transition matrix depends not only on the moduli but also on the complex arguments of $x$.
Since the above bound depends only on the absolute values, it is in general not sharp. Nonetheless, it is optimal for the self-dual
Z-invariant Ising model on isoradial graphs.

\begin{remark}
Note that finiteness of $\FinG$ was not important in our computations. Since the transition matrix is defined locally 
for each vertex, we only used the fact that all vertices have finite degree. Hence, one can consider transition matrices and 
Kac--Ward operators on infinite graphs as automorphisms of the Hilbert space $\ell^2$ on the directed edges of $\FinG$. 
The results from this section translate directly to this setting by interchanging all maxima with suprema. 
This is used in \cite{Lis} to analyse infinitely dimensional Kac--Ward operators.
\end{remark}

\subsection{High and low-temperature spectral radii} 
We will now use the bounds from the previous section in a more concrete setting of the high and low-temperature weight systems.
We define 
\begin{align*}
R(\beta) =\sup_{\FinG} \rho\big[\TrM_{\FinG}(\tanh \beta J)\big] \quad \text{and} \quad R^*(\beta) =\sup_{\FinG} \rho\big[\TrM_{\FinG^{*}}(\exp(-2 \beta J))\big],
\end{align*}
where the suprema are taken over all finite subtilings of $\InfG$. 
The reason for our particular choice of the high and low-temperatures regimes in the statement of Theorem \ref{thm:FreeEnergy} is the following result:
\begin{theorem} \label{thm:boundedradius}
If the coupling constants satisfy
\begin{itemize}
\item[(i)] the high-temperature condition, then $\sup_{\beta \in K} R(\beta) < 1$ for any compact set $K \subset \Bhigh$.
\item[(ii)] the low-temperature condition, then $\sup_{\beta \in K} R^*(\beta) < 1$ for any compact set $K \subset \Blow$.
\end{itemize}
\end{theorem}
\begin{proof} 
We will prove part (i). Fix a compact set $K \subset \Bhigh$ and let
\begin{align*}
L(\beta) = \sup_{j \in [m,M]} \frac{ |\tanh \beta j| }{ \tanh j} =  
\sup_{j \in [m,M]} \frac{\cosh j}{\sinh j}\sqrt{\frac{ \cosh (2j \RP \beta) - \cos (2j \IP \beta) }{\cosh (2j \RP \beta) + \cos (2j \IP \beta) }} .
\end{align*}
By compactness of $[m,M]$ and the fact that the hyperbolic tangent does not vanish and is continuous in the right half-plane, $L$
is a continuous function on $\{\beta: 0 <\RP \beta  \}$.
From a simple computation, it follows that $L(\beta)<1$ if and only if
\begin{align*}
\cosh (2j \RP \beta)/\cosh 2j < \cos (2j\IP \beta) \qquad \text{for all} \ j \in [m,M].
\end{align*}
The above inequality can hold only when $|\RP \beta| <1$ and when the right hand side is positive. 
The latter is in particular true when $2M|\IP \beta| < \frac{\pi}{2}$. 
Under these assumptions, both sides of the inequality are decreasing functions of~$j$.
This means that the above condition is satisfied whenever $0<\RP \beta <1$, $2M|\IP \beta| < \frac{\pi}{2}$ and
$\cosh (2m \RP \beta)/\cosh 2m < \cos (2M\IP \beta)$.
Hence, by the definition of $\Bhigh$, we have that $\Bhigh \subset \{\beta: L(\beta) <1 \}$ and thus, by continuity of $L$, 
\begin{align*}
s:=\sup_{\beta \in K} L(\beta)<1.
\end{align*} 
From the definition of $L$, it follows that
\begin{align*}
|\tanh \beta J_{\ue}|/ \tanh J_{\ue}\leq s \qquad \text{for all} \ \ue \in E(\InfG) \ \text{and} \ \beta \in K.
\end{align*}

We assume that the coupling constants $J$ satisfy the high-temperature condition, which means that the weight system $\tanh J$ factorizes to a contractive
weight system $\dx$. Therefore, for $\beta \in K$, $\tanh \beta J$ 
factorizes to a weight system $\dx(\beta)$ satisfying $|\dx_{\ve}(\beta)| = \sqrt{|\tanh \beta J_{\ue}| / \tanh J_{\ue}}\cdot|\dx_{\ve}|$, and hence $|\dx_{\ve}(\beta)|^2/s \leq  |\dx_{\ve}|^2$
for all $\ve \in \dE(\InfG)$. Since $\arctan$ is increasing and $\dx$ is contractive, we have by Corollary \ref{cor:contraction} that 
$\| \TrM_{\FinG}(\dx(\beta))\| \leq s$ for all subtilings~$\FinG$ and all $\beta \in K$. The claim follows because the spectral radius is bounded
from above by the operator norm, and $\TrM_{\FinG}(\dx(\beta))$ has the same spectral radius as $\TrM_{\FinG}(\tanh \beta J)$.

Part (ii) involves less computations and can be proved similarly after noticing that 
\[
\Blow =\Big\{\beta: \sup_{j \in [m,M]} \frac{|\exp(-2\beta j)|}{\exp ( -2 j)} <1\Big\}.   \qedhere
\]
\end{proof}
In the light of Thoerem \ref{thm:expansions} and the remarks which follow it, we are now in a position to prove our main result.

\section{Proof of Theorem \ref{thm:FreeEnergy}} \label{sec:freeenergy}

\begin{proof}
We will prove part (i). Suppose that the coupling constants satisfy the high-temperature condition 
and fix a compact set $K \subset \Bhigh$.
We have to show that the functions $f^{\free}_{\FinG}$ extend analytically to $\Bhigh$ and are uniformly bounded on $K$.

First of all, since zero is not in $\Bhigh$, the factor $1/\beta$ is analytic on $\Bhigh$ and uniformly bounded on $K$. Thus, it is enough 
to consider functions of the form $\ln \ZZ^{\free}_{\FinG}(\beta) /|V(\FinG)|$.
We will use the formula from part (i) of Theorem~\ref{thm:expansions}.
The logarithm of the partition function can therefore be expressed as a sum of three different terms. 
The first one is the constant $|V(\FinG)| \ln 2$, which equals $ \ln 2$ after rescaling by the number of vertices.

To talk about the second term, which comes from the product of hyperbolic cosines, one has to argue that
there is a continuous branch of $\ln(\cosh \beta J_e)$ on $\Bhigh$. Indeed, one can take the principal 
value of the logarithm since $\RP (\cosh \beta J_e )= \cosh (J_e \RP \beta) \cos (J_e \IP \beta)> 0$ on $\Bhigh$. Analyticity of this term follows
since $ \cosh \beta J_e $ is analytic. Furthermore, 
we have 
\begin{align*}
\Big| \ln \Big(\prod_{\ue \in E(\FinG)} \cosh \beta J_e\Big) \Big| &\leq \sum_{\ue \in E(\FinG)} \big| \ln (\cosh \beta J_e) \big| \\
&\leq \sum_{\ue \in E(\FinG)} \Big( \big|\ln |\cosh \beta J_e|\big|  + |\text{Arg}(\cosh \beta J_e)| \Big)  \\
&\leq  | E(\FinG)|\Big( \sup_{j\in[m,M]} \big|\ln |\cosh \beta j|\big|   + \pi/2\Big).
\end{align*} 
Since the hyperbolic cosine does not vanish in the right half-plane and $[m,M]$ is compact,
the above supremum is a continuous function of $\beta$ on $\Bhigh$, and therefore is bounded on $K$.
The number of edges is bounded by the number of vertices times the maximal degree of $\InfG$, and thus,
after rescaling by the volume, this term is uniformly bounded in $\FinG$.

The last term is given by the logarithm of the determinant of the
Kac--Ward operator. Let $\lambda_k$, $k\in\{1,2,\ldots,2n\}$, $n=|E(\FinG)|$, be the eigenvalues of $\TrM_{\FinG}(\tanh \beta J)$.
By Theorem \ref{thm:boundedradius}, we know that their moduli are bounded from above by some constant $s<1$ (uniformly in $\FinG$ and $\beta \in K$).
One can therefore define the logarithm by its power series around one, i.e.\
\begin{align*}
\ln  \det \big[\Id - \TrM_{\FinG}(\tanh \beta J)\big] &= \ln  \prod_{k=1}^{2n} (1- \lambda_k)  = \sum_{k=1}^{2n} \ln (1- \lambda_k) \\
 &= -\sum_{k=1}^{2n} \sum_{r=1}^{\infty} \lambda_k^r/r = - \sum_{r=1}^{\infty} \sum_{k=1}^{2n} \lambda_k^r/r \\ 
 &= - \sum_{r=1}^{\infty} \trace [\TrM^r_{\FinG}(\tanh \beta J) ]/r,
\end{align*}
where $\trace$ is the trace of a matrix. It is clear that $\trace [\TrM^r_{\FinG}(\tanh \beta J)]$ is an analytic function of $\beta$.
Moreover, $|\trace [\TrM^r_{\FinG}(\tanh \beta J)]| \leq 2|E(\FinG)| s^r$ for any~$r$,
and therefore the above series converges uniformly on $K$. 
It follows that the series defines a holomorphic function on $\Bhigh$. 
Again, after rescaling by the number of vertices, it becomes uniformly bounded in $\FinG$. 
This completes the proof of the first part of the theorem.

The proof of part (ii) uses the second formula from Theorem \ref{thm:expansions} and proceeds in a similar manner. 
\end{proof}

\textbf{Acknowledgments.} 
The research was supported by NWO grant Vidi 639.032.916.

\bibliographystyle{amsplain}
\bibliography{criticalpoint}
\end{document}